\documentclass[11pt]{amsart}
\usepackage{a4wide,amsmath,amssymb,graphicx}

\theoremstyle{plain}
\newtheorem{theorem}{Theorem}

\begin{document}
\renewcommand{\baselinestretch}{1.15}

\markboth{M.~Baake and H.~K\"{o}sters}{Random point 
sets and their diffraction}

\title{Random point sets and their diffraction}

\author{Michael Baake and Holger K\"{o}sters}
\address{Fakult\"{a}t f\"{u}r Mathematik, 
      Universit\"{a}t Bielefeld, Postfach 100131, 
     33501 Bielefeld, Germany}

\begin{abstract}
  The diffraction of various random subsets of the integer lattice
  $\mathbb{Z}^{d}$, such as the coin tossing and related systems, are
  well understood. Here, we go one important step beyond and consider
  random point sets in $\mathbb{R}^{d}$. We present several systems
  with an effective stochastic interaction that still allow for
  explicit calculations of the autocorrelation and the diffraction
  measure. We concentrate on one-dimensional examples for illustrative
  purposes, and briefly indicate possible generalisations to higher
  dimensions.

  In particular, we discuss the stationary Poisson process in
  $\mathbb{R}^{d}$ and the renewal process on the line. The latter
  permits a unified approach to a rather large class of
  one-dimensional structures, including random tilings.  Moreover, we
  present some stationary point processes that are derived from the
  classical random matrix ensembles as introduced in the pioneering
  work of Dyson and Ginibre. Their re-consideration from the
  diffraction point of view improves the intuition on systems with
  randomness and mixed spectra.
\end{abstract} 

\maketitle
\thispagestyle{empty}

% \begin{keywords}   diffraction spectra;  
% stochastic point processes; random matrix ensembles.
% \end{keywords}\bigskip

\section{Introduction}

Mathematical diffraction theory is an abstraction of kinematic
diffraction \cite{C} that is both mathematically rich and practically
relevant. Its insight helps to explore the true setting and difficulty
of the inverse problem of structure determination from (kinematic)
diffraction data; compare \cite{BG07,BG09,GB} and references therein.
Moreover, mathematical diffraction theory has several important
connections with the theory of dynamical systems \cite{BL,BL2,BLM,BW}.

While this is well studied in the classic case of crystals and, more
generally, pure point diffractive systems, much less is known about
structures with continuous (or mixed) diffraction spectra. The
Thue-Morse chain is a well-known example with singular continuous
diffraction, but mixed dynamical spectrum; see \cite{BG08,HB,ME} for
more.  Recently, some progress has also been made for point sets of
stochastic origin; see \cite{BBM} and references therein for a survey.

Nevertheless, the collection of fully worked out and understood
examples is relatively meagre in comparison with the situation of pure
point diffractive systems. Here, we summarise two important and
versatile examples from \cite{BBM} and augment them with two examples
from the theory of random matrices. The latter essentially derive from
old papers by Dyson \cite{D} and Ginibre \cite{G}, which were later
re-analysed by Mehta \cite{M}, though the results seem unknown in
diffraction theory.

Let us briefly recall the setting of mathematical diffraction theory;
see \cite{Hof,BBM} and references therein for more. The underlying
structure is modelled by an essentially translation bounded measure
$\omega$, which may be signed or even complex; see \cite{RS} for
background on measure theory. The corresponding \emph{autocorrelation}
measure $\gamma$, or autocorrelation for short, is defined as a
volume-weighted limit,
\begin{equation} \label{eq:def-auto}
   \gamma \, = \lim_{R\to\infty} 
   \frac{\omega^{}_{R} * \widetilde{\omega^{}_{R}}}
   {\mathrm{vol} (B_{R})} \, ,
\end{equation}
where $\omega^{}_{R}$ denotes the restriction of $\omega$ to the (open)
ball of radius $R$ around the origin. Moreover, $\widetilde{\omega}$
is the `flipped-over' measure defined by $\widetilde{\omega} (g)
= \overline{\omega (\widetilde{g})}$ for arbitrary continuous functions
of compact support, with $\widetilde{g} (x) = \overline{g(-x)}$. We
implicitly use the Riesz-Markov representation theorem that allows us
to identify regular Borel measures with linear functionals on the
space of continuous functions with compact support. In general, the
limit in \eqref{eq:def-auto} need not exist, but we will only consider
situations where it does, at least almost surely in the probabilistic
sense.

By construction, $\gamma$ is a positive definite measure, hence it is
always Fourier transformable. Here, we follow the convention of
\cite{BBM}, with the factor $2\pi$ in the exponent (via $e^{-2\pi i k
  x}$) rather than in front of the integral.  The result is the
\emph{diffraction} $\widehat{\gamma}$, which is a positive measure (by
the Bochner-Schwartz theorem \cite{RS}) that describes, loosely
speaking, how much intensity is scattered into any given volume of
space; see \cite{C,Hof} for more. The diffraction measure has a unique
decomposition
\begin{equation} \label{eq:decomp}
  \widehat{\gamma} = \bigl( \widehat{\gamma} \bigr)_{\mathsf{pp}}
  + \bigl( \widehat{\gamma} \bigr)_{\mathsf{sc}}
  + \bigl( \widehat{\gamma} \bigr)_{\mathsf{ac}}
\end{equation}
into its \emph{pure point} part (which is a countable sum of
Dirac measures, known as Bragg peaks), its \emph{absolutely
continuous} part (which comprises everything that can be expressed
by a locally integrable density relative to Lebesgue measure, known
as diffuse scattering) and its \emph{singular continuous} part
(which is everything that remains, and is often disregarded in
crystallography -- but see \cite{W}). 

The focus of this paper is on systems with structural disorder, which
means that we will see either purely absolutely continuous spectra or
mixtures thereof with pure point spectra. We begin with a review of
the Poisson process and the renewal process, where we follow
\cite{BBM} and adapt it to the concrete setting of crystallography.
Then, we consider some point processes that can be extracted from
random matrix theory and give rise to further examples that are
explicitly computable.

\section{Poisson process}

The homogeneous Poisson process in $\mathbb{R}^{d}$ is an ergodic
point process that is often considered as a model for an ideal gas. If
$\rho$ denotes its (point) density, the process is characterised
\cite{F1,F2,BBM} by the two requirements that the number of points in
a (measurable) set $A \subset \mathbb{R}^{d}$ is Poisson distributed
with parameter $\rho \lambda(A)$, where $\lambda$ is Lebesgue measure,
and that the number of points in sets $A_{1} , A_{2} , \ldots , A_{m}$
are independent random variables, for any collection of pairwise
disjoint, measurable subsets of $\mathbb{R}^{d}$. The Poisson process
is a model for an ideal gas of pointlike particles.

Let us consider such a process, with density $\rho$.  Due to
ergodicity, almost every realisation of it possesses a natural
autocorrelation, and the latter can be calculated via the Palm measure
of the process \cite[Thm.~3]{BBM}. The result reads
\begin{equation} \label{eq:P-auto}
    \gamma^{}_{P} = \rho \, \delta^{}_{0} + \rho^2  \lambda  ,
\end{equation}
where $\delta_{0}$ is the normalised Dirac measure at $0$. With
$\widehat{\delta_{0}}= \lambda$ and $\widehat{\lambda}=\delta_{0}$,
one obtains 
\begin{equation} \label{eq:P-diff}
    \widehat{\gamma^{}_{P}} =
    \rho^2  \delta^{}_{0} + \rho \, \lambda  ,
\end{equation}
which is the diffraction measure. We skip the proof, but mention that,
in one dimension, the result can easily be derived from the renewal
theorem discussed in Section~\ref{sec:ren}, with an exponential
waiting time distribution \cite{BBM}. Apart from the trivial point
measure at $0$, the diffraction measure $\widehat{\gamma^{}_{P}}$ is
absolutely continuous.

An interesting modification emerges from a \emph{marked} Poisson
process, where each point of a given realisation randomly gets the
weight $1$ or $-1$ with equal probability.  This leads to the
following modification of Eqs.~\eqref{eq:P-auto} and
\eqref{eq:P-diff}.

\begin{theorem}
  Consider a typical realisation of the homogeneous Poisson process of
  density $\rho$ in $\mathbb{R}^{d}$, which is a simple point set
  $\varLambda \subset \mathbb{R}^{d}$. Let $\omega =
  \sum_{x\in\varLambda} W_{\! x}\, \delta_{x}$ be a random Dirac comb
  where\/ $(W_{x})^{}_{x\in\varLambda}$ constitutes an i.i.d.\ family
  of random variables that take the values $1$ and $-1$ with equal
  probability. Then, the corresponding autocorrelation and diffraction
  measures almost surely read
\[
    \gamma^{}_{\omega} = \rho\, \delta_{0}
    \quad \text{and}\quad 
    \widehat{\gamma^{}_{\omega}} = \rho\,\lambda .
\]
\end{theorem}
\begin{proof}
  This is the situation of the random weight model of
  \cite[Ex.~7]{BBM}, applied to a stationary Poisson process, which is
  an ergodic and simple point process. The result now follows from
  \cite[Thm.~4 and Cor.~1]{BBM} by a small calculation.
\end{proof}

When the density is $\rho=1$, this is one of many examples with
diffraction measure $\lambda$, which include the coin tossing sequence
on the integer lattice and the Rudin-Shapiro sequence \cite{HB,BG09},
but also various dynamical systems of algebraic origin \cite{BW} such
as Ledrappier's shift on $\mathbb{Z}^{2}$.  This provides ample
evidence that the inverse problem of structure determination becomes
significantly more involved in the presence of diffuse scattering.

One limitation of the Poisson process for applications in physics is
the missing uniform discreteness. This can be overcome by an
additional hard-core condition, as in the classic Mat\'{e}rn process;
see \cite{SS} and references therein for a formulation that matches
our setting. Here, a realisation of a homogeneous Poisson process is
randomly marked and then thinned out on the basis of a pairwise
comparison up to a certain distance. This leads to a modified point
set that is uniformly discrete (meaning that the minimum distance
between any two points is a positive number). Despite this
modification, autocorrelation and diffraction can still be calculated
explicitly; see \cite{BBM,Kai} and references therein for more.

\section{Renewal process}\label{sec:ren}

The situation of random point sets is significantly simpler in one
dimension, because a large class of processes can be characterised
constructively as a renewal process. Here, one starts from a
probability measure $\mu$ on $\mathbb{R}_{+}$ (the positive real line)
and considers a machine that moves at constant speed along the real
line and drops points on the line with a waiting time that is
distributed according to $\mu$.  Whenever this happens, the internal
clock is reset and the process resumes. Let us (for simplicity) assume
that both the velocity of the machine and the expectation value of
$\mu$ are $1$, so that we end up with realisations that are, almost
surely, point sets in $\mathbb{R}$ of density $1$ (after we let
the starting point of the machine move to $-\infty$).

Clearly, the process just described defines a stationary process.  It
can thus be analysed by considering all realisations which contain the
point $0$.  Moreover, there is a clear (distributional) symmetry
around this point, so that we can determine the autocorrelation (in
the sense of \eqref{eq:def-auto}) of almost all realisations from
studying what happens to the right of $0$. Indeed, if we want to know
the frequency per unit length of the occurrence of two points at
distance $x$ (or the corresponding density), we need to sum the
contributions that $x$ is the first point after $0$, the second point,
the third, and so on. In other words, we almost surely obtain the
autocorrelation
\begin{equation} \label{eq:auto-1}
    \gamma \; = \; \delta_{0} + \nu + \widetilde{\nu}
\end{equation}
with $\nu = \mu + \mu * \mu + \mu * \mu
* \mu + \ldots$ and $\widetilde{\nu}$ as defined above, where
the proper convergence of the sum of iterated convolutions follows
from \cite[Lemma~4]{BBM}. Note that the point measure at $0$ simply
reflects that the almost sure density of the resulting point set is
$1$. Indeed, $\nu$ is a translation bounded positive measure, and
satisfies the renewal relations (see \cite[Ch.~XI.9]{F2} or 
\cite[Prop.~1]{BBM} for a proof)
\begin{equation}\label{eq:ren-rel}
   \nu \; = \; \mu + \mu * \nu \qquad\text{and}\qquad
   (1-\widehat{\mu}\, )\, \widehat{\nu} 
    \; = \; \widehat{\mu}\, ,
\end{equation}
where $\widehat{\mu}$ is a uniformly continuous and bounded function
on $\mathbb{R}$. Note that the second equation emerges from the first
by Fourier transform, but has been rearranged to indicate why the set
$\{k \mid \widehat{\mu} (k) = 1 \}$ will become important below.  In
this setting, the measure $\gamma$ of \eqref{eq:auto-1} is both
positive and positive definite.

Based on the structure of the support of the underlying
probability measure $\mu$, one can now formulate the
following result for the diffraction of the renewal process.
\begin{theorem}
  Let $\mu$ be a probability measure on $\mathbb{R}_{+}$ with mean
  $1$, and assume that a moment of $\,\mu$ of order $1+\varepsilon$
  exists for some $\,\varepsilon > 0$. Then, the point sets obtained
  from the stationary renewal process based on $\mu$ almost surely
  have a diffraction measure of the form
\[
    \widehat{\gamma} \; = \; \bigl( \widehat{\gamma} 
    \bigr)_{\mathsf{pp}}  + (1-h)\,\lambda ,
\]
  where $h$ is a locally integrable function on $\mathbb{R}$
  that is continuous except for at most countably many
  points. It is given by
\[
    h(k) \; = \; \frac{2\,\bigl(\lvert\widehat{\mu} 
    (k)\rvert^2 - \mathrm{Re} (\widehat{\mu}(k))\bigr)}
    {\lvert 1 - \widehat{\mu} (k)\rvert^2}\, .
\]
   Moreover, the pure point part is given by
\[
   \bigl( \widehat{\gamma} \bigr)_{\sf pp} \; =
   \begin{cases} \delta^{}_{0} , & \text{if\/ $\mathrm{supp} 
     (\mu)$ is not a subset of a lattice}, \\
       \delta^{}_{\mathbb{Z}/b} , & \text{if\/ $b\mathbb{Z}$ is the
       coarsest lattice that contains $\mathrm{supp} (\mu)$.}
     \end{cases}
\]
\end{theorem}

\begin{proof}
  The process has a well-defined autocorrelation $\gamma$ as outlined
  above and given in Eq.~\eqref{eq:auto-1}. Due to the ergodicity of
  the process, this means that almost every realisation of the process
  is a simple point set with this autocorrelation. Since $\gamma$
  is positive definite, it is Fourier transformable, with
  $\widehat{\gamma}$ being a positive measure on $\mathbb{R}$.

  The point measure at $0$ (which is always present) reflects the fact
  that the resulting point set almost surely has density $1$; see
  \cite[Thm.~1]{BBM} and its proof for a detailed argument.  The
  functional form of $h$ can be calculated from the second renewal
  relation in \eqref{eq:ren-rel} whenever one has $\widehat{\mu}
  (k)\ne 1$.  Its local integrability everywhere is a consequence of
  the assumed moment condition, by an application of
  \cite[Thm.~1.5.4]{U}.

  The distinction via the nature of $\mathrm{supp} (\mu)$ takes
  care of the set $\{ k\in\mathbb{R} \mid \widehat{\mu} (k)= 1 \}$,
  which is either $\{ 0 \}$ or countable.  The result now follows from
  \cite[Lemma~5 and Thm.~1]{BBM} together with Remark~3 of the same
  paper.
\end{proof}

The renewal process is a versatile method to produce point sets on the
line. These include random tilings with finitely many intervals (which
are Delone sets) as well as the homogeneous Poisson process on the
line (where $\mu$ is the exponential distribution with mean $1$); see
\cite[Sec.~3]{BBM} for explicit examples and applications.

\section{Random matrix ensembles and random point sets 
on the line}

The global eigenvalue distribution of random orthogonal, unitary or
symplectic matrix ensembles is known to asymptotically follow the
classic semi-circle law. More precisely, this law describes the
eigenvalue distribution of the underlying ensembles of symmetric,
Hermitian and symplectic matrices with Gaussian distributed entries.
The corresponding random matrix ensembles are called GOE, GUE and GSE,
with attached $\beta$-parameters $1$, $2$ and $4$, respectively. They
permit an interpretation as a Coulomb gas, where $\beta$ is the power
in the central potential; see \cite{AGZ,M} for general background and
\cite{D} for the results that are relevant to our point of view.

{}For matrices of dimension $N$, the semi-circle has radius
$\sqrt{2N/\pi}$ and area $N$. Note that, in comparison to \cite{M}, we
have rescaled the density by a factor $1/\sqrt{\pi}$ here, so that we
really have a semi-circle, and not a semi-ellipse. To study the local
eigenvalue distribution with our application in mind, we rescale the
central region (between $\pm 1$, say) by $\sqrt{2N/\pi} $. This leads,
in the limit as $N\to\infty$, to a new ensemble of point sets on the
line that can be interpreted as a stationary, ergodic point process of
intensity $1$; for $\beta=2$, see \cite[Ch.~4.2]{AGZ} or \cite{Sos}
and references therein for details.  Since the process is simple
(meaning that, almost surely, no point is occupied twice), almost all
realisations are point sets of density $1$.

It is possible to calculate the autocorrelation of these processes, on
the basis of Dyson's correlation functions \cite{D}.  Though the
latter originally apply to the circular ensembles, they have been
adapted to the other ensembles by Mehta \cite{M}.  For all three
ensembles mentioned above, this leads to an autocorrelation of the
form
\begin{equation}\label{eq:Dyson-auto}
    \gamma = \delta_{0} + 
    \bigl( 1 - f(\lvert x \rvert ) \bigr) \lambda 
\end{equation}
where $f$ is a locally integrable function that depends on $\beta$.
Defining $s (r) = \frac{\sin (\pi r)}{\pi r}$, one obtains (with $r\ge
0$)
\begin{equation}\label{eq:f}
   f(r) = \begin{cases}
    s(r)^{2} + s' (r) \int_{r}^{\infty} s(t) \,\mathrm{d} t , &
    \text{if $\beta = 1$}, \\
    s(r)^{2} ,  &  \text{if $\beta=2$}, \\
    s(2r)^2 - 2 s' (2r) \int_{0}^{r} s(2t) \,\mathrm{d} t ,  &
    \text{if $\beta = 4$} .
    \end{cases}
\end{equation}

The diffraction measure is the Fourier transform of $\gamma$, which
has also been calculated in \cite{D,M}. Observing
$\widehat{\delta_{0}} =\lambda$ and $\widehat{\lambda} = \delta_{0}$,
  the result is always of the form
\begin{equation}\label{eq:Dyson-Fourier}
   \widehat{\gamma} = \delta_{0} + 
   \bigl( 1 - b(k) \bigr) \lambda =
   \delta_{0} + h(k) \, \lambda ,
\end{equation}
where $b = \widehat{f}$. The Radon-Nikodym density $h$ for
$\beta=1$ reads
\begin{equation}\label{eq:Radon-1}
   h_{1} (k) = \begin{cases}
   \lvert k \rvert \bigl(2 - \log(2\lvert k \rvert + 1)\bigr), &
   \text{if $\lvert k \rvert \le 1$}, \\
   2 - \lvert k \rvert \, \log 
   \frac{2 \lvert k \rvert + 1}{2 \lvert k \rvert - 1} , &
   \text{if $\lvert k \rvert > 1$} ,
    \end{cases}
\end{equation}
where $k\in\mathbb{R}$. The result for $\beta=2$ is simpler
and reads
\begin{equation}\label{eq:Radon-2}
   h_{2} (k) = \begin{cases}
   \lvert k \rvert, &
   \text{if $\lvert k \rvert \le 1$}, \\
   1 , &
   \text{if $\lvert k \rvert > 1$} ,
    \end{cases}
\end{equation}
while $\beta = 4$ leads to
\begin{equation}\label{eq:Radon-4}
   h_{4} (k) = \begin{cases}
   \frac{1}{4}\lvert k \rvert \bigl( 2 - \log
   \bigl| 1 - \lvert k \rvert \bigr| \bigr) , &
   \text{if $\lvert k \rvert \le 2$}, \\
   1 , &
   \text{if $\lvert k \rvert > 2$} .
    \end{cases}
\end{equation}
Figure~\ref{fig:dyson} illustrates the three cases. To summarise:

\begin{figure}
\centerline{\includegraphics[width=\textwidth]{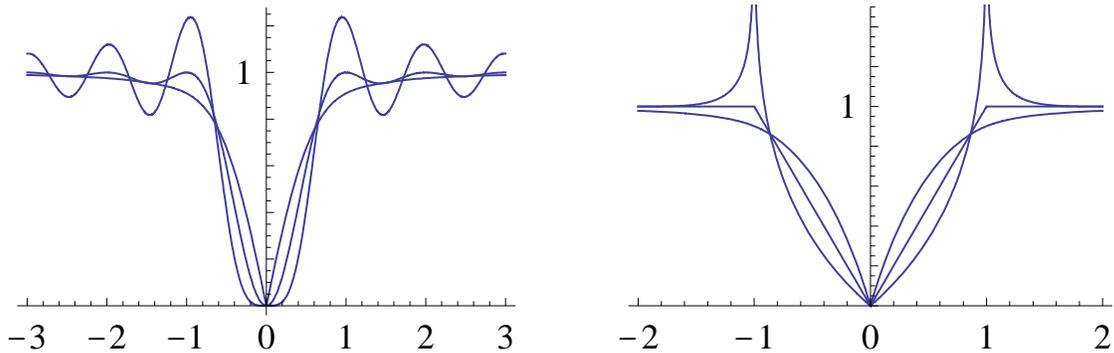}}
\caption{Absolutely continuous part of the autocorrelation (left) and
  the diffraction (right) for the three point set ensembles on the
  line, with $\beta\in\{1,2,4\}$. On the left, the oscillatory
  behaviour increases with $\beta$. On the right, $\beta=2$
  corresponds to the piecewise linear function with bends at $0$ and
  $\pm 1$, while $\beta=4$ shows a locally integrable singularity at
  $\pm 1$. The latter reflects the slowly decaying oscillations on
  the left.  \label{fig:dyson} }
\end{figure}

\begin{theorem}
  The eigenvalues of the Dyson random matrix ensembles for 
  $\beta \in \{1,2,4\}$, in the scaling of the local region around\/
  $0$ as used above, almost surely give rise to point sets of density
  $1$, with autocorrelation and diffraction measures as specified in
  Eqs.~\eqref{eq:Dyson-auto} and \eqref{eq:Dyson-Fourier}.  \hfill
  $\square$
\end{theorem}

Note that $h_{4}$ is smooth at $k=\pm 2$, but has integrable
singularities at $k= \pm 1$.  The latter are a consequence of the
stronger oscillatory behaviour of the function $f_{4}$ at integer
values, as was already noticed in \cite{D}. When extrapolating to
other values of $\beta$ (in particular to $\beta > 4$), this is the
onset of another Bragg peak.

It is well-known that the circular random matrix ensembles (COE, CUE,
CSE) asymptotically give rise to the same local correlations
\cite{D,M}, and hence to the same autocorrelation and diffraction
(after appropriate rescaling).

\section{Random matrix ensembles and random point sets 
in the plane}

The above examples were derived from matrix ensembles with real
eigenvalues, and thus lead to point processes in $\mathbb{R}$.  There
is also one ensemble, due to Ginibre \cite{G} (see also \cite{M}), of
general complex matrices with Gaussian distributed entries that will
give rise to a stationary point process in $\mathbb{R}^{2}$. Again,
this emerges (by proper rescaling) from the eigenvalues (now seen as
elements of the plane), which approach uniform distribution in a
circle of radius $\sqrt{N/\pi}$ (and hence area $N$) as $N\to\infty$.

As before, the system can be interpreted as a Coulomb gas, with a
potential parameter $\beta = 2$. Other matrix ensembles permit this
interpretation, too, but do not seem to correspond to interesting
stationary processes, wherefore we stick to Ginibre's example here.

Following the original approach of \cite{G}, the limit $N\to\infty$
leads to a stationary and ergodic, simple point process of intensity
$1$, so that almost every realisation is a point set in the plane of
density $1$. Using complex variables $z_{i}\in\mathbb{C}\simeq
\mathbb{R}^{2}$, the $2$-point correlation function is of
determinantal form,
\begin{equation}\label{eq:Gin-corr}
   \rho (z^{}_{1}, z^{}_{2} ) = 
   e^{- \pi ( \lvert z^{}_{1}\rvert^{2} + \lvert z^{}_{2}\rvert^{2})}\,
   \left| \begin{matrix} e^{\pi \lvert z^{}_{1} \rvert^{2}} &
   e^{ \pi z^{}_{1}\, \overline{\! z^{}_{2} \! } \,} \\ 
   e^{\pi \, \overline{\! z^{}_{1}\! }\,  z^{}_{2}} &
   e^{\pi \lvert z^{}_{2} \rvert^{2}} \end{matrix}\right| =
   \bigl( 1 - e^{- \pi \lvert  z^{}_{1} - z^{}_{2}\rvert^{2}} \bigr),
\end{equation}
see \cite{G} or \cite{M} for a derivation. Note that, despite using
complex coordinates here, the expression is calculated relative to the
volume element of real coordinates (hence relative to Lebesgue
measure, as in \cite{M}).  The result is translation invariant and
only depends on the distance $r$ between the two points.

As a consequence, the autocorrelation of a realisation almost surely
reads
\begin{equation}\label{eq:Gin-auto}
   \gamma = \delta_{0} +  (1 - e^{- \pi r^{2}} )\, \lambda ,
\end{equation}
which is radially symmetric, with $r$ as above. By a standard
calculation, the Fourier transform of $\gamma$ results in
\begin{equation}\label{eq:Gin-Fourier}
   \widehat{\gamma} = \delta_{0} + 
   \bigl( 1 -  e^{- \pi \lvert k \rvert^{2}}\bigr) \lambda ,
\end{equation}
so that we obtain a self-dual pair of measures under Fourier transform
(as in the Poisson process of density $1$).  The radial dependence is
illustrated in Figure~\ref{fig:gin}.

\begin{figure} 
\centerline{\includegraphics[width=0.6\textwidth]{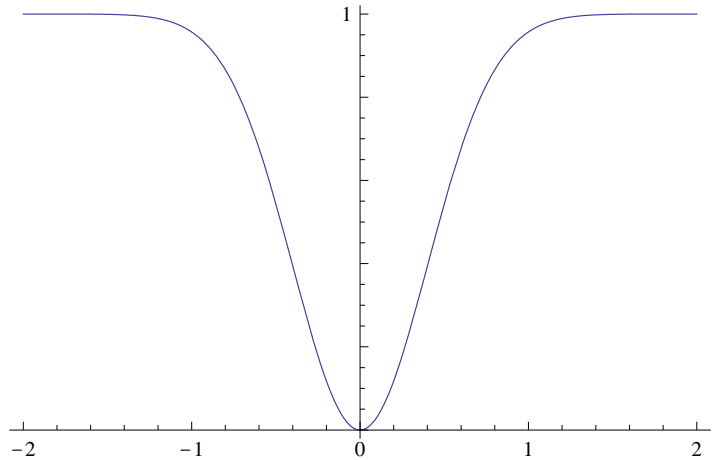}}
\caption{Radial dependence of the absolutely continuous part both of
  the autocorrelation and the diffraction measure for the planar point
  set ensemble, as derived from Ginibre's matrix ensemble.
   \label{fig:gin} }
\end{figure}
\begin{theorem}
  The Ginibre complex random matrix ensemble, in the scaling used
  above, almost surely results in point sets of density $1$, with
  autocorrelation \eqref{eq:Gin-auto} and diffraction
  \eqref{eq:Gin-Fourier}.  \hfill $\square$
\end{theorem}

\section{Summary and Outlook}

In this short communication, we have discussed several explicit
examples of stochastic point sets with explicitly computable
autocorrelation and diffraction measures. The viewpoint of point
process theory provides a universal platform to do so, though our
examples above also admit a direct approach. It would be interesting
to extend this to a family of processes, with $\beta$ as parameter in
the spirit of \cite{DE,VV,KS}, which then interpolates between the
Poisson process of density $1$ ($\beta = 0$) and the integer lattice
(which is approached as $\beta \to\infty$).

An interesting question concerns the connection with dynamical systems
theory, in particular the general relation between diffraction and
dynamical spectra.  Recent progress suggests that such a connection
might indeed exist, although it will certainly be more involved than
in the pure point diffractive case.

One fundamental shortcoming so far is the lack of understanding and
explicit examples for randomness with interaction. A first step in
that direction needs the inclusion of Gibbs measures for equilibrium
states, though it is not clear at the moment to what extent one can
derive explicit examples (such as the classic and well-known Ising
lattice gas).

\section*{Acknowledgements}
MB would like to thank P.\ Forrester and O.\ Zeitouni for
discussions on the Coulomb gas and the random matrix
ensembles. This work was supported by the German Research
Council (DFG), within the CRC 701.

\end{document}